\documentclass[12pt,letterpaper]{article}
\usepackage{amsmath}
\usepackage{amssymb}
\usepackage{enumitem}
\usepackage{mathtools}
\usepackage{amsthm}
\usepackage[utf8]{inputenc}
\usepackage[mathscr]{euscript}
\usepackage{pst-node}
\usepackage{stackengine}
\usepackage{fullpage}
\usepackage[semicolon,authoryear]{natbib}
\usepackage[colorlinks=true,linkcolor=blue,citecolor=blue]{hyperref}
\usepackage{booktabs,tabularx}

\usepackage{tikz-cd}
\usepackage{bm}
\usepackage{geometry}
\usepackage[T1]{fontenc}
\usepackage{titling}
\usepackage{threeparttable}

\usepackage{tcolorbox}
\usepackage{listings}
\usepackage{longtable}

\usepackage{xargs}
\usepackage[colorinlistoftodos,prependcaption,textsize=tiny]{todonotes}
\newcommandx{\sgcmt}[2][1=]{\todo[linecolor=red,backgroundcolor=red!25,bordercolor=red,#1]{#2}}

\definecolor{DarkGreen}{RGB}{30,180,30}

\newif\ifbackmatter%
\newcommand{\backmatter}{\global\backmattertrue}%

\usepackage[title]{appendix}%

\theoremstyle{definition}

\theoremstyle{remark}

\newtheorem*{claim}{{\bf \emph{Lemma}}}
\newtheorem*{remark}{Remark}

\tikzcdset{
arrow style=tikz,
diagrams={>={Straight Barb[scale=0.8]}}
}

\DeclareMathOperator\sym{sym}

\newcommand{\interior}[1]{%
  {\kern0pt#1}^{\mathrm{o}}%
}

\title{Magnetic influence on ion transport in concentrated solid solutions: An analytic investigation}
\author{Timothy Carlson \& Sanjay Govindjee\\
University of California, Berkeley\\
\texttt{timothy.carlson@berkeley.edu, s\_{}g@berkeley.edu}}
\begin{document}

\maketitle

\begin{abstract}

It is well established that magnetic fields have a significant effect on transport in certain classes of electronic conductors. Less reported, however, are similar effects in solid ionic conductors. Despite the rarity of Hall mobility measurements in ionic conductors,
%significant or otherwise,
recent experimental work
in batteries and other systems
has demonstrated that an applied magnetic field can significantly and beneficially alter ionic transport and electrochemical processes in solid materials in a way that would not be predicted from na\"{\i}ve Hall coefficient estimates. In this work, the influence of a magnetic field on ion transport in solids is investigated analytically, and general multi-component transport equations accounting for magnetic effects are presented. Specific models are then derived for solid, isotropic binary and single ion conductors. Material property combinations for which magnetic field influence may become significant are then computed for certain systems subject to compositional constraints.
%such as site constrained lattices or materials with constant composition. 
Finally, it is demonstrated that the derived model for binary conductors in a magnetic field fits experimental magneto-resistance data well for the fluoride ion conducting solid Pb$_{0.66}$Cd$_{0.34}$F$_2$, provided an assumption of near degenerate multi-component transport. %Further insights into auxiliary implications of the near degenerate conduction model are also assessed through Finite Element Simulation. 
    
\end{abstract}

\tableofcontents

\section{Introduction}\label{a:c}

For many materials, it is well established that magnetic fields can significantly influence the conduction current (see, e.g., \cite{Mansfield:78} and \cite{DESIRE:24}). However, the same phenomena are rarely explored in the context of solid ionic conductors given that charge carriers are classical particles that often conduct at far smaller speeds at very high concentrations. Therefore, the Hall coefficient, when computed as $R_h=(Fcz)^{-1}$, where $F$ is Faraday's constant, $z$ is the charge of the carrier, and $c$ is the mobile carrier concentration, is typically negligible in solid ion conductors \citep[see, e.g.,][]{Carvalho:22}. Despite this, the transport of ions in solids has been shown to be sensitive to magnetic field effects in a way that significantly alters the performance of, for example, an all solid state battery. Specifically, increased ion conductivity is reported in garnet type LLZO solid electrolytes by \cite{zhang:22} when exposed to a 960\,mT magnetic field in the direction of the applied current, in addition to significantly improved cyclability of Li/LLZO/Li symmetric cells. More demonstrative magnetic field effects are obtained by \cite{Kim:25}, where measured effective conductivity is improved by over 100 percent for a viscoelastic, ionogel solid electrolyte.  Despite these observations, we know of no published rigorous analytical theories that explain them.  At best, one finds conjectures regarding magneto-hydrodynamic effects which are either
impossible or improbable in the systems under consideration.

The thermodynamics of multi-component transport in concentrated solutions have been discussed at length in the literature (see, e.g., \cite{Onsager:31} and \cite{nb:21}), and the general process of deriving useful transport equations in such systems is well established \citep[see, e.g.,][]{LI:21}. However, much of the existing literature neglects couplings to the magnetic field and, when such couplings are discussed, useful expressions relating the flux of species to driving forces are stated only in very general terms \citep[see, e.g.,][]{MAZUR:53}. Previous, related, analytical work focused on constitutive relations for species fluxes in the presence of a magnetic field has mostly been centered on ionized gas or plasma in the absence of external driving forces (see, e.g., \citet[][pp.199-201]{Balakrishnan:2021} and \cite{Kinouchi:1988}). Magnetic influence on conduction current in semi-conductors has of course been treated extensively \citep[see, e.g.,][]{Dill:25}, but species interaction is rarely considered in these cases and more pronounced magnetoresistive effects are typically better described through the use of alternative models \citep[see, e.g.,][]{christensen:24}. Similar attempts at stating general transport equations from the derived Onsager relations have also been explored \citep[see, e.g.,][]{Ayansiji:20}, but models are typically limited to a single species or are thermodynamically incompatible with more general multi-component systems. 

Multi-component transport with nontrivial driving electromagnetic fields has also been treated in the context of ionized plasma \citep[see, e.g.,][]{BRAGINSKII:65}, but transport equations are not clearly stated for a modern audience\footnote{For example, CGS units are used for magnetic fields, and the force flux relations are not clearly stated or derived in the more general cases.}. In particular, these treatments are not immediately applicable to the transport of ions in solids influenced by a magnetic field \citep[see, e.g.,][]{CARLSON:26}. Finally, previous works (see, e.g., \cite{costa:21} and \cite{Luo:22}) have summarized the physics of
magnetic fields on ion transport, 
but such attempts have not 
looked at the mathematical structure of the conductivity tensor in a general setting.
As such, a complete, analytical, thermodynamically consistent treatment of magnetic influence on transport in concentrated, multi-component solid solutions is seemingly absent from the literature.

In this work, we generalize and refine the transport relations shown in an earlier study \citep{CARLSON:26}, utilizing concentrated solution theory via the Onsager-Stefan-Maxwell formalism as presented by \cite{nb:21}, to determine regimes where magnetic influence on the transport of ions in solids may be non-negligible. In doing so, the full Lorentz force on each charged particle must be accounted for, which can be computed in a straight-forward way from the Langevin equation for each carrier \citep[see, e.g., ][pp.199-201]{Balakrishnan:2021}. From this, useful expressions relating to the flux of each species can be derived. After stating the general equations, three prototypical, isotropic cases are considered: a single ion conductor, a single ion conductor with a secondary neutral species, and a binary solid electrolyte. For the second and third case, specific analytic examples are considered where a single effective conductivity tensor can be derived: that of a site-constrained single ion conductor (ions and neutral vacancies) and a binary ionic conductor with constrained composition. The second case is potentially applicable to certain classes of Garnet-like solid electrolytes where vacancies can conceivably be modeled as a charged, secondary species \citep[see, e.g.,][]{bai:25}. Under this assumption, it is justifiable to assume that $\mu_1=-\mu_2$ for $\mu_i$ defined as the part of the chemical potential associated with the activity of species $i$ only. Finally, the example of constrained composition is applied to the super fluoride ion conductor Pb$_{0.66}$Cd$_{0.34}$F$_2$. Subsequent analysis demonstrates how the model may serve as an explanation of observed magneto-resistance and deviation from the typical Hall effect-derived conduction model, provided the acceptance of an assumption of near degenerate, multi-component transport.

\section{Theoretical preliminaries}

Some important theoretical results relating to species transport in a continuum when electromagnetic fields are first stated. More detail can be found in \cite{CARLSON:26} and \cite{kovetz:00}. Of primary importance to this work are the electric field $\bm{e}$ [V/m] and magnetic field $\bm{b}$ [T];\footnote{Often, $\bm{b}$ is called the magnetic flux density, and the current potential $\bm{h}$ is referred to as the magnetic field. These names are commonly used interchangeably given the electromagnetic duality/aether relations \citep[see, e.g.,][]{kovetz:00}.} indeed, the electromagnetic duality equations ensure that Maxwell's equations can be satisfied through knowledge of $\bm{b}$ and $\bm{e}$ alone. For each of these fields, there exist potentials, $\bm{a}$ and $\phi$, such that
\begin{align}
    \bm{b}=\nabla\times \bm{a}, \quad\text{and}\quad \bm{e}=-\bm{a}_{,t}-\nabla\phi.\label{eq:emagpots}
\end{align}
However, it is straightforward to show using a space-time deformation at constant velocity that $\bm{e}$ and the potential ${\phi}$ are not classical-observer invariant, although $\bm{b}$ and $\bm{a}$ are. Letting $\bm{v}$ denote the velocity of
moving charged matter
%a point $p$ 
relative to a rest/relativistic inertial frame, Galilean invariant electric fields and electric potentials can be constructed \citep[see, e.g.,][]{kovetz:00}. These invariant quantities, given by
\begin{align}
    \mathbb{E}_s=\bm{e}+\bm{v}\times\bm{b},&& \overline{\phi}=\phi-\bm{v}\cdot\bm{a}\label{eq:polarizer}
\end{align}
represent the fields and potentials in a co-moving frame; heuristically, these are the fields that are "felt" by the moving
material 
%located at $p$
, and are therefore the fields that influence electrochemical processes, for example. It is then possible to show that $\mathbb{E}_s$ can be written in terms of the potentials as
\begin{align*}
    \mathbb{E}_s=-\mathcal{L}_{\bm{v}}\bm{a}-\nabla\overline{\phi}
\end{align*}
where $\mathcal{L}_{\bm{v}}\bm{a}=\dot{\bm{a}}+(\nabla\bm{v})^T\bm{a}$ denotes the Lie derivative of $\bm{a}$ along $\bm{v}$.

Using the first and second laws of thermodynamics, it is then possible to derive constitutive relations for material responses to, for example, $\mathbb{E}_s$, $\bm{b}$, $\overline{\phi}$, and species concentrations $c_i$ [mol/m$^3$] by supposing the existence of a specific free energy function $\varphi$ [J/kg]. These constitutive relations then enable the definition of the free charge conduction current $\mathbb{J}_s$ [C/m$^2$s], species flux $\bm{j}_i$ [mol/m$^2$s], and electro\emph{magnetic}-chemical potential $\mu_i^{em}$ [J/mol]. Relevant to this work are the following results
%of this derivation 
\citep[see, e.g.,][]{CARLSON:26}:
\begin{align}
    \mu_i^{em}&=\rho\frac{\partial\varphi}{\partial c_i}\label{eq:chempot1}\\
    \begin{bmatrix}
        \mathbb{J}_s\\
        \bm{j}_i \\  
    \end{bmatrix} &= -\begin{bmatrix}
        \bm{\kappa} & \mathscr{C}_{\mathbb{J}_s\mu_{i}} \\
        \mathscr{C}_{\bm{j}_i\mathbb{E}_s} & \bm{M}_i  \\
    \end{bmatrix} \begin{bmatrix}
        \mathcal{L}_{\bm{v}}\bm{a} \\
        \nabla \mu_i^{em} \\
    \end{bmatrix}\label{eq:cup1},
\end{align}
where $\rho$ is the density of the material [kg/m$^3$], $\bm{\kappa}$ is the conductivity [m/N\,s], $\bm{M}_i$ denote partial mobilities, and the $\mathscr{C}$ denote electromagnetic-chemical cross couplings. Note that the electrochemical potential is typically considered to be additively decomposable into an electric part and a chemical part \cite[see e.g.,][Ch. 3]{nb:21}. However, in the presence of a magnetic field, the electric part becomes electromagnetic to preserve invariance.
In order to correctly represent the potential experienced by a given species, we define an electromagnetic-chemical potential
\begin{align}
    \mu_i^{em}=\mu_i+Fz_i\overline{\phi}\label{eq:echemdef},
\end{align}
where $F$ is Faraday's constant, $\mu_i$ is the purely chemical potential, and $z_i$ is the charge associated species $i$. Observe in a neutral species, one has $\mu_i^{em}=\mu_i$. 

When establishing the precise forms of the tensors in \eqref{eq:cup1}, it is useful to note the relation $\mathbb{J}_s=\sum_{i=1}^nFz_i\bm{j}_{i}$, where $n$ is the number of mobile species. This allows one to collapse the transport equations above into relationships involving only the species fluxes. These transport equations are then representative 
%for equations of motion for 
of the
the average velocity of a species given some driving forces; consequently, specification of particular forms for the transport tensors relevant to a given material is often more easily accomplished by first constructing the inverse relation, as those tensors are often more easily measured/computed. This is the starting point of the next section.

\section{Generalized transport in an electromagnetic field}\label{s:2}

Given the short relaxation times and large number of collisions expected for particles migrating through a solid lattice or viscoelastic gel, balancing the average force on mobile particles 
\begin{align*}
   \bm{f}^\mathrm{em}_i + \bm{f}^\mathrm{chem}_i + \sum_{j=1}^n \bm{M}'_{ij}\bm{u}_j = \bm{0},
\end{align*}
where $c_i$ [mol/m$^3$] is a measure of concentration,\footnote{In practice, it may sometimes be necessary to use a "mobile" carrier concentration rather than the total species concentration \citep[see, e.g.,][]{thompson:15}.}
$\bm{f}_i^\mathrm{em}$ is the electromagnetic Lorentz force on species $i$, 
$\bm{f}_i^\mathrm{chem} = c_i\nabla\mu_i$ is the chemical driving force on species $i$,
$\bm{u}_j$ is the average velocity of species $j$ relative to the host solid, and $\bm{M}_{ij}'$ are modified friction or interaction tensors \citep[see, e.g.,][p, 259]{nb:21}, which are proportional to blocks of the inverse of the transport relations given in \eqref{eq:cup1}.\footnote{The $\bm{M}_{ij}'$ tensors are usually given in terms of diffusion coefficients $\mathscr{D}_{ij}$ and, in the isotropic case, are written $M'_{ij}=\frac{RTc_ic_j}{c_T\mathscr{D}_{ij}}$ for $i\neq j$ and $M'_{ij}=-\sum_{j}\frac{RTc_ic_j}{c_T\mathscr{D}_{ij}}$ otherwise, where $c_T$ is the total concentration of the solution.}

In the presence of an electromagnetic field, the  Lorentz force on the particle relative to the host lattice must be accounted for when computing the force on a particle distribution. This is given by \citep[see, e.g.,][]{kovetz:00}
\begin{align*}
 \bm{f}_i^\mathrm{em}=   Fz_ic_i(\bm{e}+(\bm{v}+\bm{u}_i)\times\bm{b})=Fz_ic_i(\mathbb{E}_s+\bm{u}_i\times\bm{b}).
\end{align*}
Summing this with the chemical force then leads to the extended equation
\begin{align}
    c_i\left(-Fz_i\bm{u}_i
    \times\bm{b}+Fz_i\mathscr{L}_{\bm{v}}\bm{a}+\nabla\mu_i^{em}\right)=\sum_{j=1}^n \bm{M}'_{ij}\bm{u}_j,\label{eq:intdif}
\end{align}
where it is noted that the electric potential is absorbed into the electromagnetic-chemical potential.
It is then useful to rewrite this equation in terms of the relative molar flux $\bm{j}_i=c_i\bm{u}_i$ [mol/m$^2$\,s] as
\begin{align}
    Fz_i\mathscr{L}_{\bm{v}}\bm{a}+\nabla\mu_i^{em}=\sum_{j=1}^n \overline{\bm{L}}_{ij}\bm{j}_j,\label{eq:intdif2}
\end{align}
%given that we must solve for the fluxes, 
where
\begin{align}
    \overline{\bm{L}}_{ij}=\begin{cases}
        \frac{\bm{M}'_{ij}}{c_ic_j}+\frac{Fz_i}{c_i}\;\bm{\epsilon}\cdot\bm{b}&\text{if}\quad i=j\\[4pt]
        \frac{\bm{M}_{ij}'}{c_ic_j}&\text{otherwise},
    \end{cases}
\end{align}
and $\bm{\epsilon}$ is the third order Levi-Civita tensor. 
%This system of equations can then be solved alongside Gauss's law and the Maxwell-Amp\`ere equation to solve for the electromagnetic-chemical potentials, electrostatic potential, and magnetic potential of the system \citep[see, e.g.,][]{CARLSON:26}. \sgcmt{This last sentence seems un-needed for this paper.}

%In the most general case, $N$ mobile charged species may interact with $M$ uncharged species, and so the system can be written in matrix form as
%\begin{align}
%    \begin{bmatrix}
%        Fz_1\mathscr{L}_{\bm{v}}\bm{a}+\nabla\mu_1^{em}\\
%        ...\\
%        Fz_N\mathscr{L}_{\bm{v}}\bm{a}+\nabla\tilde{\mu}^e_N\\
%        \nabla\tilde{\mu}^e_{N+1}\\
%        ...\\
%        \nabla\tilde{\mu}^e_{N+M}
%    \end{bmatrix} = \begin{bmatrix}
%        \overline{L}_{11} & ... & \overline{L}_{1N} & \overline{L}_{1(N+1)} & ... & \overline{L}_{1(N+M)}\\
%        ... & ... & ... & ... & ... & ... \\
%        \overline{L}_{N1} & ... & \overline{L}_{NN} & \overline{L}_{N(N+1)} & ... & \overline{L}_{N(N+M)}\\
%        \overline{L}_{(N+1)1} & ... & \overline{L}_{(N+1)N} & \overline{L}_{(N+1)(N+1)} & ... & \overline{L}_{(N+1)(N+M)}\\
%        ... & ... & ... & ... & ... & ... \\
%        \overline{L}_{(N+M)1} & ... & \overline{L}_{(N+M)N} & \overline{L}_{(N+M)(N+1)} & ... & \overline{L}_{(N+M)(N+M)}
%    \end{bmatrix}\begin{bmatrix}
%        \bm{j}_1\\
%        ...\\
%        \bm{j}_N\\
%        \bm{j}_{N+1}\\
%        ...\\
%        \bm{j}_{N+M}
%    \end{bmatrix}
%\end{align}
%after re-indexing so that the charged species appear first. 

\begin{remark}
If each $\bm{M}'_{ii}$ is invertible, then each $\overline{\bm{L}}_{ii}$ is also. While tedious, it is straightforward to show using the Sherman–Morrison–Woodbury relation \citep[see e.g.,][]{hager89} that for a material that exhibits time-reversal symmetry in the absence of magnetic influence (i.e. $\bm{M}_{ii}'$ is symmetric), we must have 
\begin{align}
    \overline{\bm{L}}_{ii}^{-1}=\frac{c_i^2\bm{M}_{ii}'^{-1}}{1+F^2z_i^2c_i^2(\bm{M}_{ii}':\bm{b\otimes\bm{b}})/(\det(\bm{M}_{ii}'))}\left[\bm{I}-Fz_ic_i(\bm{\epsilon}\cdot\bm{b})\bm{M}_{ii}'^{-1}+\frac{F^2z_i^2c_i^2}{\det(\bm{M}_{ii}')}\bm{M}_{ii}'\bm{b}\otimes\bm{b}\right].
\end{align}
Thus, the magnetic field breaks so-called time reversal symmetry. Further, in the most general case $\bm{M}'_{ii}$ need only be positive definite, so one can replace $\bm{M}'_{ii}$ with $\sym(\bm{M}'_{ii})$ and $\bm{b}$ with $\bm{b}+\bm{\omega}$ for the most general possible expression, where $\bm{\omega}$ is defined such that $\bm{\epsilon}\cdot(Fz_ic_i\bm{\omega})$ is the antisymmetric part of $\bm{M}'_{ii}$. Thus, the diagonal blocks are still invertible via the formula given above for any second-law satisfying material. However, the mathematical simplifications of the next sections, which are made possible by the algebraic properties of the set $\{\bm{I},\bm{\epsilon}\cdot\bm{b},\bm{b}\otimes\bm{b}\}$, cannot be utilized in the general case. 

It should also be noted that magnetic influence on the interaction tensors $\bm{M}'_{ij}$ themselves has not been considered, as the magnetic field would have to meaningfully alter momentum exchange from collisions to influence this parameter. Although possible in principle, the magnetic field strength required to meaningfully alter the diffusivity tensor in liquids has been shown to be similar in magnitude to the field produced by a neutron star \citep[see e.g,][]{Kinouchi:1988}.
\end{remark}

\section{Isotropic conduction models}\label{s:3}

In this section, we cover three specific models of general interest: a true single carrier conductor, a single ion conductor with a mobile secondary species, and a binary conductor. Using the methodology presented in the derivation of each case and the theory presented in the previous sections, more general models can be derived, such as the case of a binary electrolyte with a tertiary, mobile neutral species. 

In each of the following derivations, the transport tensors are assumed to be isotropic. Additionally, each of the three lemmas in Appendix \ref{a:a} will be utilized heavily throughout the remainder of this section.

\subsection{Single ion conductor}\label{s:3.1}

For a single, non-interacting charged species in an isotropic medium, \eqref{eq:intdif2} simplifies to
\begin{align}
    c\left(-Fz\bm{u}\times\bm{b}+Fz\mathscr{L}_{\bm{v}}\bm{a}+\nabla{\mu}^{em}\right)= M'\bm{u}.
\end{align}
Further, for such a material it is also the case \citep[see e.g.,][p. 260]{nb:21} that $M'=-F^2c^2z^2\rho$, where $\rho$ is the material's electric resistivity, and so
\begin{align}
   Fz\mathscr{L}_{\bm{v}}\bm{a}+\nabla{\mu}^{em}= -(F^2z^2c\rho)\bm{u}+ Fz\bm{u}\times\bm{b}=-Fz\rho\mathbb{J}_s+\frac{1}{c}\mathbb{J}_s\times \bm{b},
\end{align}
after noting that the average conduction current is defined as $\mathbb{J}_s=Fcz\bm{u}$. Solving for $\mathbb{J}_s$ then leads to the standard \citep[see e.g.,][]{chien:80} expression for conductivity accounting for the Hall effect, 
\begin{align}
    \mathbb{J}_s&=-\frac{\rho}{\rho^2+R_h^2\bm{b}\cdot\bm{b}}\left[\bm{I}+\frac{R_h}{\rho}\bm{\epsilon}\cdot\bm{b}+\frac{R_h^2}{\rho^2}\bm{b}\otimes\bm{b}\right]\left(\mathscr{L}_{\bm{v}}\bm{a}+\frac{1}{Fz}\nabla{\mu}^{em}\right),\label{eq:SIC}
\end{align}
where $R_h=\frac{1}{Fcz}$. Note that the carrier concentration in many materials is large enough that magneto-resistance can be neglected. 

\subsection{Single ion conductor with a secondary mobile species}

In the case of a material with one mobile charged species but two mobile, interacting species, it follows from \eqref{eq:intdif} that
\begin{align}
    \begin{bmatrix}
        Fz_1\mathscr{L}_{\bm{v}}\bm{a}+\nabla\mu_1^{em}\\
        \nabla\mu_2
    \end{bmatrix}=\begin{bmatrix}
        \bm{M}_{11}'/c_1+Fz_1\bm{\epsilon}\cdot\bm{b} & \bm{M}_{12}'/c_1\\
        \bm{M}_{21}'/c_2 & \bm{M}_{22}'/c_2
    \end{bmatrix}\begin{bmatrix}
        \bm{{u}}_1\\\bm{{u}}_2
    \end{bmatrix}.
\end{align}
Note that we consider $\mu_2$ here given that $z_{2}=0$. Then, from the definition of the relative molar flux $\bm{j}_i=c_i\bm{{u}}_i$, we have
\begin{align*}
    \begin{bmatrix}
        Fz_1\mathscr{L}_{\bm{v}}\bm{a}+\nabla\mu_1^{em}\\
        \nabla\mu_2
    \end{bmatrix}&=\begin{bmatrix}
        \bm{M}_{11}'/c_1^2+\frac{Fz_1}{c_1}\bm{\epsilon}\cdot\bm{b} & \bm{M}_{12}'/c_1c_2\\
        \bm{M}_{21}'/c_1c_2 & \bm{M}_{22}'/c_2^2
    \end{bmatrix}\begin{bmatrix}
        \bm{j}_1\\\bm{j}_2
    \end{bmatrix}\\&=-\begin{bmatrix}
        Fz_1\bm{I} & \bm{0}\\
        \bm{0} &F\bm{I}
    \end{bmatrix}\begin{bmatrix}
        \rho_{11}\bm{I}-R_h^{(1)}\bm{\epsilon}\cdot\bm{b} & \rho_{12}\bm{I}\\
        \rho_{12}\bm{I} & \rho_{22}\bm{I}
    \end{bmatrix}\begin{bmatrix}
        Fz_1\bm{I} & \bm{0}\\
        \bm{0} &F\bm{I}
    \end{bmatrix}\begin{bmatrix}
        \bm{j}_1\\\bm{j}_2
    \end{bmatrix}\\
    &=-\bm{Z}\bm{\rho}\bm{Z}\begin{bmatrix}
        \bm{j}_1\\\bm{j}_2
    \end{bmatrix}\,,
\end{align*}
where the second relation holds for materials that can be assumed to be isotropic conductors. Here, $\rho_{11}=-M'_{11}/F^2c_1^2z_1^2$, $\rho_{12}=\rho_{21}=M'_{12}/F^2c_1c_2z_1$, and $\rho_{22}=-M'_{22}/F^2c_2^2$ are partial resistivies corresponding to each relative transport mechanism and $R_h^{(1)}=\frac{1}{Fz_1c_1}$ is the partial Hall coefficient representing the individual interaction of the charged species with an external magnetic field. It then follows from Lemma 3 (Appendix \ref{a:a}) that the conductivity $\bm{\kappa}=\bm{\rho}^{-1}$ is a block symmetric matrix, where each block can be written explicitly as
\begin{align}
    \bm{\kappa}_{11}&=\frac{\rho_{22}\kappa_D}{1+(\kappa_D\rho_{22}R_h^{(1)})^2(\bm{b}\cdot\bm{b})}\left[\bm{I}+\kappa_D\rho_{22}R_h^{(1)}\bm{\epsilon}\cdot\bm{b}+(\kappa_D\rho_{22}R_h^{(1)})^2\bm{b}\otimes\bm{b}\right]\label{eq:Kc11}\\
    \bm{\kappa}_{12}&=\frac{-\rho_{12}\kappa_D}{1+(\kappa_D\rho_{22}R_h^{(1)})^2(\bm{b}\cdot\bm{b})}\left[\bm{I}+\kappa_D\rho_{22}R_h^{(1)}\bm{\epsilon}\cdot\bm{b}+(\kappa_D\rho_{22}R_h^{(1)})^2\bm{b}\otimes\bm{b}\right]\label{eq:Kc12}\\
    \bm{\kappa}_{22}&=\frac{\rho_m\kappa_D}{1+(\kappa_D\rho_{22}R_h^{(1)})^2(\bm{b}\cdot\bm{b})}\left[\bm{I}+\kappa_D\rho_{22}R_h^{(1)}\bm{\epsilon}\cdot\bm{b}+(\kappa_D\rho_{22}R_h^{(1)})^2\bm{b}\otimes\bm{b}\right]+\frac{1}{\rho_{22}}\bm{I}\label{eq:Kc22}
\end{align}
for
\begin{align*}
    \rho_m=\frac{\rho_{12}^2}{\rho_{22}},\qquad\kappa_D=(\rho_{11}\rho_{22}-\rho_{12}^2)^{-1}.
\end{align*}
Therefore, we have
\begin{align}
    \begin{bmatrix}
        \mathbb{J}_s\\[4pt] F\bm{j}_2
    \end{bmatrix}=-\bm{\kappa}\begin{bmatrix}
        \mathscr{L}_{\bm{v}}\bm{a}+\frac{1}{z_1F}\nabla\mu_1^{em}\\[4pt]
        \frac{1}{F}\nabla\mu_2
    \end{bmatrix}\label{eq:SIC2}
\end{align}
since the conduction current $\mathbb{J}_s$ is directly proportional to  the flux of the single charged species via $\mathbb{J}_s=Fz_1\bm{j}_1$. 

\subsection{Binary conductor}\label{s:3.3}

In the case of a material with two mobile, charged, interacting species, it follows from \eqref{eq:intdif} that
\begin{align}
    \begin{bmatrix}
        Fz_1\mathscr{L}_{\bm{v}}\bm{a}+\nabla\mu_1^{em}\\
        Fz_2\mathscr{L}_{\bm{v}}\bm{a}+\nabla\mu_2^{em}
    \end{bmatrix}=\begin{bmatrix}
        \bm{M}_{11}'/c_1+Fz_1\bm{\epsilon}\cdot\bm{b} & \bm{M}_{12}'/c_1\\
        \bm{M}_{21}'/c_2 & \bm{M}_{22}'/c_2+Fz_2\bm{\epsilon}\cdot\bm{b}
    \end{bmatrix}\begin{bmatrix}
        \bm{{u}}_1\\\bm{{u}}_2
    \end{bmatrix}.
\end{align}
Then, from the definition of the relative molar flux $\bm{j}_i=c_i\bm{{u}}_i$, we have
\begin{align*}
    \begin{bmatrix}
        Fz_1\mathscr{L}_{\bm{v}}\bm{a}+\nabla\mu_1^{em}\\
        Fz_2\mathscr{L}_{\bm{v}}\bm{a}+\nabla\mu_2^{em}
    \end{bmatrix}&=\begin{bmatrix}
        \bm{M}_{11}'/c_1^2+\frac{Fz_1}{c_1}\bm{\epsilon}\cdot\bm{b} & \bm{M}_{12}'/c_1c_2\\
        \bm{M}_{21}'/c_1c_2 & \bm{M}_{22}'/c_2^2+\frac{Fz_2}{ c_2}\bm{\epsilon}\cdot\bm{b}
    \end{bmatrix}\begin{bmatrix}
        \bm{j}_1\\\bm{j}_2
    \end{bmatrix}\\&=-\begin{bmatrix}
        Fz_1\bm{I} & \bm{0}\\
        \bm{0} &Fz_2\bm{I}
    \end{bmatrix}\begin{bmatrix}
        \rho_{11}\bm{I}-R_h^{(1)}\bm{\epsilon}\cdot\bm{b} & \rho_{12}\bm{I}\\
        \rho_{12}\bm{I} & \rho_{22}\bm{I}-R_h^{(2)}\bm{\epsilon}\cdot\bm{b}
    \end{bmatrix}\begin{bmatrix}
        Fz_1\bm{I} & \bm{0}\\
        \bm{0} &Fz_2\bm{I}
    \end{bmatrix}\begin{bmatrix}
        \bm{j}_1\\\bm{j}_2
    \end{bmatrix}\\
    &=-\bm{Z}\bm{\rho}\bm{Z}\begin{bmatrix}
        \bm{j}_1\\\bm{j}_2
    \end{bmatrix}
\end{align*}
where the second relation holds for materials that can be assumed to be isotropic conductors. Here, $\rho_{11}=-M'_{11}/F^2c_1^2z_1^2$, $\rho_{12}=\rho_{21}=-M'_{12}/F^2c_1c_2z_1z_2$, and $\rho_{22}=-M'_{22}/F^2c_2^2z_2^2$ are partial resistivies corresponding to each relative transport mechanism and $R_h^{(1)}=\frac{1}{Fz_1c_1}$ and $R_h^{(2)}=\frac{1}{Fz_2c_2}$ are partial Hall coefficients representing the individual interaction of each species with an external magnetic field.

It then follows from Lemmas 1, 2, and 3 (Appendix \ref{a:a}) that $\bm{\kappa}=\bm{\rho}^{-1}$ is in this case given by
\begin{align*}
    \bm{\kappa}=\begin{bmatrix}
        \bm{D}_\kappa^{-1}(\rho_{22}\bm{I}-R_h^{(2)}\bm{\epsilon}\cdot\bm{b}) & -\rho_{12}\bm{D}_\kappa^{-1}\\
        -\rho_{12}\bm{D}_\kappa^{-1}& \bm{D}_\kappa^{-1}(\rho_{11}\bm{I}-R_h^{(1)}\bm{\epsilon}\cdot\bm{b})
    \end{bmatrix},
\end{align*}
for
\begin{align*}
    \bm{D}_\kappa&=(\rho_{11}\bm{I}-R_h^{(1)}\bm{\epsilon}\cdot\bm{b})(\rho_{22}\bm{I}-R_h^{(2)}\bm{\epsilon}\cdot\bm{b})-\rho_{12}^2\bm{I}\\
    &=[\rho_{11}\rho_{22}-\rho_{12}^2-R_h^{(1)}R_h^{(2)}(\bm{b}\cdot\bm{b})]\bm{I}-(\rho_{11}R_h^{(2)}+\rho_{22}R_h^{(1)})\bm{\epsilon}\cdot\bm{b}+R_h^{(1)}R_h^{(2)}\bm{b}\otimes\bm{b}.
\end{align*}
Straightforward calculation then yields
\begin{align*}
    \bm{D}_\kappa^{-1}=\frac{\kappa_D}{1+\kappa_D^2R_D^2(\bm{b}\cdot\bm{b})}\left[\bm{I}+\kappa_DR_D\bm{\epsilon\cdot\bm{b}}+\frac{1}{\rho_D}\left(\kappa_DR_D^2-R_\delta\right)\bm{b}\otimes\bm{b}\right]
\end{align*}
for
\begin{align*}
    R_D=\rho_{11}R_h^{(2)}+\rho_{22}R_h^{(1)},\qquad
    \rho_D=\rho_{11}\rho_{22}-\rho_{12}^2,\qquad
    R_\delta&=R_h^{(1)}R_h^{(2)},\qquad \kappa_D=(\rho_D-R_\delta(\bm{b}\cdot\bm{b}))^{-1}.
\end{align*}

In this case, both species carry a charge and therefore contribute to the conduction current. From the definition of the charge tensor $\bm{Z}$, it follows that $\bm{Z}[\bm{j_1}\quad\bm{j_2}]^T=[\bm{i}_1\quad\bm{i}_2]^T$, where $\bm{i}_1$ and $\bm{i}_2$ are the partial contributions to the current from species 1 and 2, respectively, such that $\mathbb{J}_s=\bm{i}_1+\bm{i}_2$. Therefore, we can write
\begin{align}
    \begin{bmatrix}
        \bm{i}_1\\\bm{i}_2
    \end{bmatrix}&=-\begin{bmatrix}
        \bm{D}_\kappa^{-1}(\rho_{22}\bm{I}-R_h^{(2)}\bm{\epsilon}\cdot\bm{b}) & -\rho_{12}\bm{D}_\kappa^{-1}\\
        -\rho_{12}\bm{D}_\kappa^{-1}& \bm{D}_\kappa^{-1}(\rho_{11}\bm{I}-R_h^{(1)}\bm{\epsilon}\cdot\bm{b})
    \end{bmatrix}\begin{bmatrix}
        \mathscr{L}_{\bm{v}}\bm{a}+\frac{1}{Fz_1}\nabla\mu_1^{em}\\
        \mathscr{L}_{\bm{v}}\bm{a}+\frac{1}{Fz_2}\nabla\mu_2^{em}
    \end{bmatrix}.
\end{align}
Carrying out the tensor multiplication leads to the final expression
\begin{align}
    \mathbb{J}_s&=-\bm{\kappa}_1\left(\mathscr{L}_{\bm{v}}\bm{a}+\frac{1}{Fz_1}\nabla\mu_1^{em}\right)-\bm{\kappa}_{2}\left(\mathscr{L}_{\bm{v}}\bm{a}+\frac{1}{Fz_2}\nabla\mu_2^{em}\right),\label{eq:binary}
\end{align}
for
\begin{align*}
    \bm{\kappa_1}&=\frac{\kappa_D}{1+\kappa_D^2R_D^2(\bm{b}\cdot\bm{b})}\left[[\rho_{22}-\rho_{12}+(R^{(2)}_h\kappa_D R_D)\bm{b}\cdot\bm{b}]\bm{I}+[(\rho_{22}-\rho_{12})\kappa_DR_D-R^{(2)}_h]\bm{\epsilon\cdot\bm{b}}\right.\\ &\left.+\left[\frac{\rho_{22}-\rho_{12}}{\rho_D}\left(\kappa_DR_D^2-R_\delta\right)-R_h^{(2)}\kappa_DR_D\right]\bm{b}\otimes\bm{b}\right]\\
    \bm{\kappa_2}&=\frac{\kappa_D}{1+\kappa_D^2R_D^2(\bm{b}\cdot\bm{b})}\left[[\rho_{11}-\rho_{12}+(R^{(1)}_h\kappa_D R_D)\bm{b}\cdot\bm{b}]\bm{I}+[(\rho_{11}-\rho_{12})\kappa_DR_D-R^{(1)}_h]\bm{\epsilon\cdot\bm{b}}\right.\\ &\left.+\left[\frac{\rho_{11}-\rho_{12}}{\rho_D}\left(\kappa_DR_D^2-R_\delta\right)-R_h^{(1)}\kappa_DR_D\right]\bm{b}\otimes\bm{b}\right]\,.
\end{align*}
Observe in particular that if $\rho_D$ is very small or just one of the $R^{(i)}_h$ are very large then magnetic influence can be substantial.

\section{Analytical examples of systems with constraints}\label{s:4}

In this section, we consider specific examples where a single conductivity tensor can be computed. Subsequent analysis is then performed to study the parameter combinations required to observe significant magnetic field influence. 

\subsection{Site constrained single ion conductor with vacancies}

We first consider the case of a site constrained single-ion conductor; that is, a single ion conductor where unfilled vacancies can be considered a neutral, secondary species. We further assume that $z_1=1$, as would be the case for conducting lithium ions. Assuming that vacancy and ion motion are intrinsically tied to the same, fixed lattice structure where available vacancies have no associated charge, we can assume $\bm{j_1}\approx-\bm{j_2}$. In view of equation \eqref{eq:SIC2}, we can solve for $\nabla\mu_2$ in this case, resulting in
\begin{align}
    \mathbb{J}_s&=-(\bm{\kappa}_{11}-\bm{\kappa}_{12}(\bm{\kappa}_{22}+\bm{\kappa}_{12})^{-1}(\bm{\kappa}_{11}+\bm{\kappa}_{12}))\left(\mathscr{L}_{\bm{v}}\bm{a}+\frac{1}{F}\nabla\mu_1^{em}\right).\label{eq:Jbase}
\end{align}
By plugging \eqref{eq:Kc11}, \eqref{eq:Kc12}, and \eqref{eq:Kc22} into \eqref{eq:Jbase}, straightforward calculation reveals that
\begin{align}
    \mathbb{J}_s=-\frac{\kappa_{\mathrm{eff}}}{1+\left(R_h^{(1)}\kappa_{\mathrm{eff}}\right)^2\bm{b}\cdot\bm{b}}\left[\bm{I}+R_h^{(1)}\kappa_{\mathrm{eff}}\bm{\epsilon}\cdot\bm{b}+\left(R_h^{(1)}\kappa_{\mathrm{eff}}\right)^2\bm{b}\otimes\bm{b}\right]\left(\mathscr{L}_{\bm{v}}\bm{a}+\frac{1}{F}\nabla\mu_1^{em}\right),
\end{align}
where $\kappa_{\mathrm{eff}}=(\rho_{11}-\rho_{12})^{-1}$. Structurally, this equation is identical to the typical single carrier model, \eqref{eq:SIC}, and the Hall coefficient remains unchanged. Therefore, in this constrained system, a small Hall coefficient will imply magnetic field influence on transport is negligible unless $\mu_i^{em}$ itself depends on the magnetic field. See Figure \ref{f:1} for a visual depiction of the influence of the Hall coefficient on conduction, which clearly only affects conduction in directions perpendicular to $\bm{b}$. 

\begin{figure}\centering
\includegraphics[height=0.4\textheight]{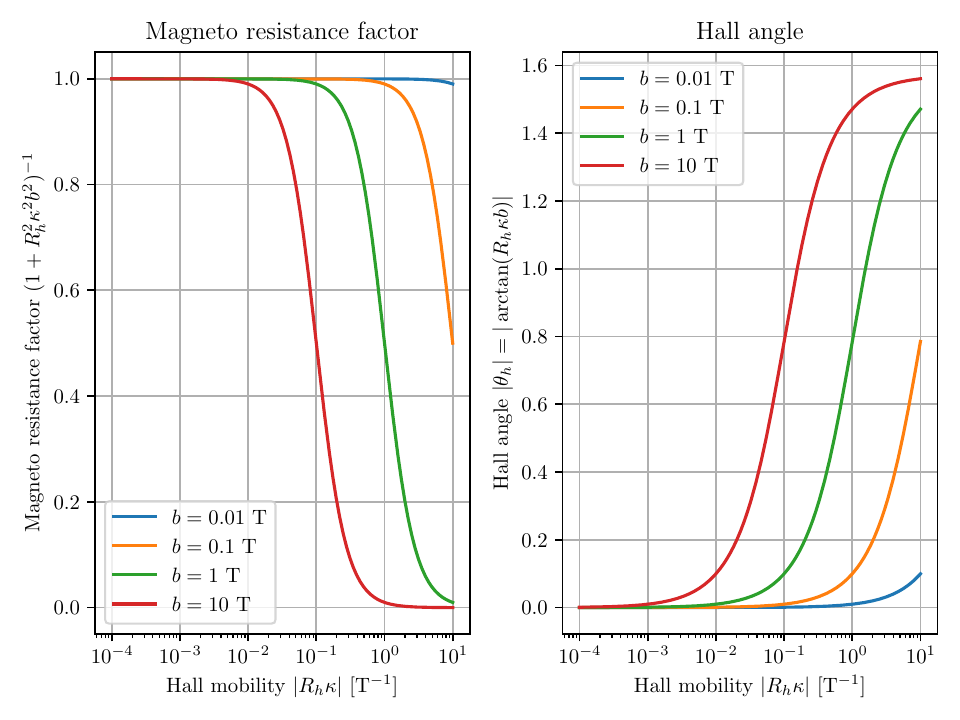}
\caption{Sensitivity of the conductivity to the Hall mobility $R_h\kappa$ at different field strengths $b=|\bm{b}|$. If $\bm{g}_\perp$ is the component of the driving force perpendicular to $\bm{b}$, the magneto resistance factor $({1+R_h^2\kappa^2\bm{b}\cdot\bm{b}})^{-1}$ is the reduction in conductivity in the direction of $\bm{g}_\perp$ (left). The Hall angle $\theta_h=\arctan({R_h^1\kappa|\bm{b}|})$ is the angle between the resultant current in the direction of $\bm{g}_\perp$ and in the direction perpendicular to both $\bm{g}_\perp$ and $\bm{b}$ (right). Note that the Hall coefficient may be negative, so the Hall angle plot is rotated 180 degrees about the origin when Hall mobility is negative.} \label{f:1}
\end{figure}

\subsection{Binary conductor with constrained composition}\label{s:4.2}

We now focus our attention on binary conductors with a compositional constraint forcing $\frac{1}{F{z_1}}\nabla\mu_1^{em}=\frac{1}{F{z_2}} \nabla\mu_2^{em}$. Recalling that $\mu_1^{em}=Fz_1\overline{\phi}+\mu_1$, it is immediate that this case applies to materials with constant composition/activity since $\frac{1}{F{z_1}}\nabla\mu_1^{em}=\nabla\overline{\phi}$. Therefore, an effective conductivity can be determined by simply adding $\bm{\kappa}_1$ and $\bm{\kappa_2}$. That is,
\begin{align}
    \mathbb{J}_s=-\bm{\kappa}_{\mathrm{eff}}\left(\mathscr{L}_{\bm{v}}\bm{a}+\frac{1}{Fz_1}\nabla\mu_1^{em}\right)
\end{align}
with
\begin{align*}
    \bm{\kappa}_{\mathrm{eff}}&=\frac{\kappa_D}{1+\kappa_D^2R_D^2(\bm{b}\cdot\bm{b})}\left\{\left[\rho_{11}-2\rho_{12}+\rho_{22}+\kappa_DR_D\left(R_h^{(1)}+R_{h}^{(2)}\right)(\bm{b}\cdot\bm{b})\right]\right.\bm{I}\\
    &+\left[\kappa_DR_D(\rho_{11}-2\rho_{12}+\rho_{22})-(R_h^{(1)}+R_h^{(2)})\right]\bm{\epsilon}\cdot\bm{b}\\
    &+\left.\left[\frac{\rho_{11}-2\rho_{12}+\rho_{22}}{\rho_{D}}\left(\kappa_DR_D^2-R_\delta\right)-\kappa_DR_D\left(R_h^{(1)}+R_h^{(2)}\right)\right]\bm{b}\otimes\bm{b}\right\}.
\end{align*}
We can then define the parameters
\begin{align*}
    \kappa_{\mathrm{eff}}&=\kappa_D(\rho_{11}-2\rho_{12}+\rho_{22})=\frac{\rho_{11}-2\rho_{12}+\rho_{22}}{\rho_{11}\rho_{22}-\rho_{12}^2-R_h^{(1)}R_h^{(2)}(\bm{b}\cdot\bm{b})}\\
    R_{\mathrm{eff}}&=\frac{R_D}{\rho_{11}-2\rho_{12}+\rho_{22}}=\frac{\rho_{11}R_h^{(2)}+\rho_{22}R_h^{(1)}}{\rho_{11}-2\rho_{12}+\rho_{22}}\\
    h_{\mathrm{mod}}&=\frac{R_h^{(1)}+R_{h}^{(2)}}{\rho_{11}-2\rho_{12}+\rho_{22}}
\end{align*}
%which allow us to more succinctly write
%\begin{align*}
%    \bm{\kappa}_{\mathrm{eff}}&=\frac{\kappa_{\mathrm{eff}}}{1+\kappa_{\mathrm{eff}}^2R_{\mathrm{eff}}^2(\bm{b}\cdot\bm{b})}\left\{\left[1+\kappa_{\mathrm{eff}}R_{\mathrm{eff}}h_{\mathrm{mod}}(\bm{b}\cdot\bm{b})\right]\bm{I}+\left[\kappa_{\mathrm{eff}}R_{\mathrm{eff}}-h_{\mathrm{mod}}\right]\bm{\epsilon}\cdot\bm{b}\right.\\
%    &+\left.\left[\kappa_{\mathrm{eff}}R_{\mathrm{eff}}\left(\frac{R_D}{\rho_D}-h_{\mathrm{mod}}\right)+\frac{R_\delta}{\rho_D}\right]\bm{b}\otimes\bm{b}\right\}.
%\end{align*}
and projections $\bm{P}_{\bm{b}}^\perp=\bm{I}-(\bm{b}\cdot\bm{b})^{-1}\bm{b}\otimes\bm{b}$ and $\bm{P}_{\bm{b}}=(\bm{b}\cdot\bm{b})^{-1}\bm{b}\otimes\bm{b}$, where $\bm{I}=\bm{P}_{\bm{b}}^\perp+\bm{P}_{\bm{b}}$, which enable us to write the final expression for the effective conductivity as
\begin{align}
    \bm{\kappa}_{\mathrm{eff}}&=\frac{\kappa_{\mathrm{eff}}}{1+\kappa_{\mathrm{eff}}^2R_{\mathrm{eff}}^2(\bm{b}\cdot\bm{b})}\left\{\left[1+\kappa_{\mathrm{eff}}R_{\mathrm{eff}}h_{\mathrm{mod}}(\bm{b}\cdot\bm{b})\right]\bm{P}_{\bm{b}}^\perp+\left[\kappa_{\mathrm{eff}}R_{\mathrm{eff}}-h_{\mathrm{mod}}\right]\bm{\epsilon}\cdot\bm{b}\right\}+\kappa_{\mathrm{eff}}^0\bm{P}_{\bm{b}}.\label{eq:keff_con}
\end{align}
Here, $\kappa_{\mathrm{eff}}^0=\rho_D^{-1}$ is the effective conductivity when $\bm{b}=0$; if $R^{(1)}_hR_h^{(2)}\ll \rho_{11}\rho_{22}-\rho_{12}^{2}$, which is likely true even in most cases of near-degenerate transport, $\kappa_{\mathrm{eff}}^0\approx \kappa_{\mathrm{eff}}$ to excellent accuracy for achievable field strengths. It is then clear from the forms of the coefficients that, even if the Hall coefficients in a binary conductor are small, the effective hall parameters $R_{\mathrm{eff}}$ can be non-negligible in a concentrated solution with a near-degenerate $\bm{\rho}$ matrix. 

\begin{figure}\centering
\includegraphics[height=0.6\textheight]{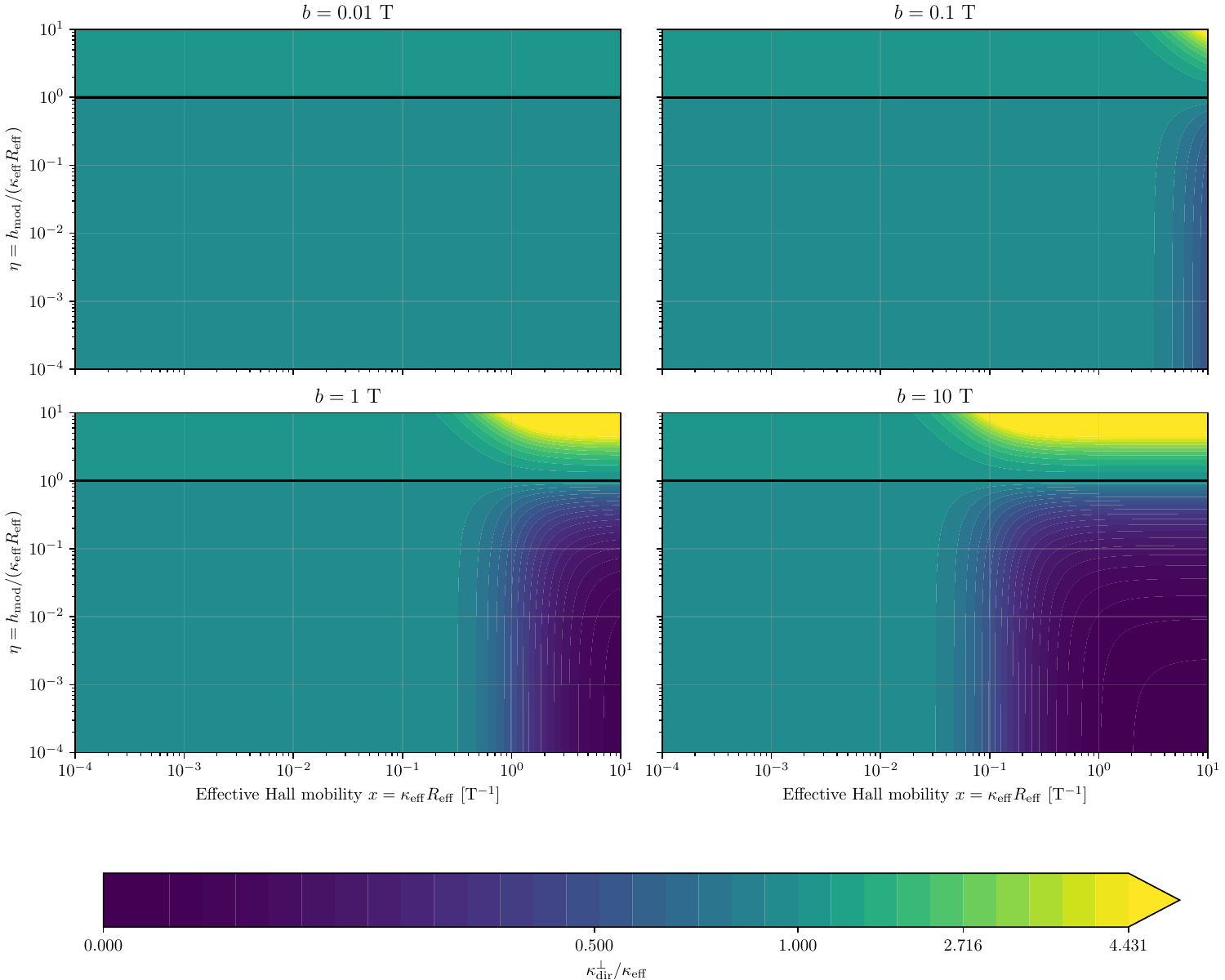}
\caption{Sensitivity of the ratio $\kappa_{\mathrm{eff}}/\kappa_{\mathrm{eff}}^\perp$, where $\kappa_{\mathrm{eff}}^\perp$ is the coefficient of $\bm{P}_b^\perp$ in the conductivity tensor, to the Hall mobility $R_{\mathrm{eff}}\kappa_{\mathrm{eff}}$ 
%\sgadd{(should these be $\kappa_\mathrm{eff}$ and $R_\mathrm{eff}$?)}
and binary coupling $\eta=h_{\mathrm{mod}}/\kappa_{\mathrm{eff}}R_{\mathrm{eff}}$ at different field strengths $b=|\bm{b}|$.} \label{f:2}
\end{figure}
\begin{figure}\centering
\includegraphics[height=0.6\textheight]{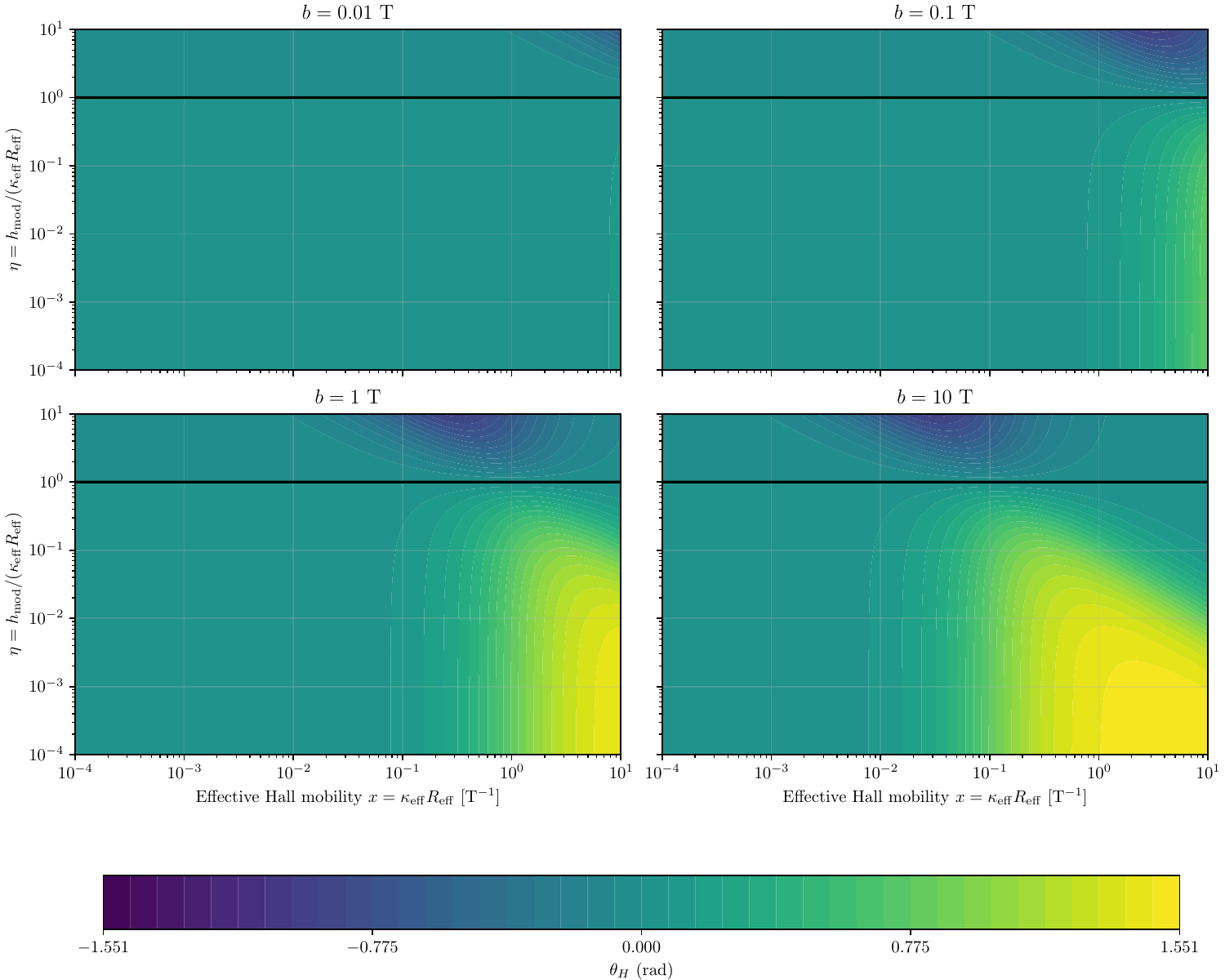}
\caption{Sensitivity of the hall angle to the effective Hall mobility $\kappa_{\mathrm{eff}}R_{\mathrm{eff}}$ and binary coupling $\eta=h_{\mathrm{mod}}/\kappa_{\mathrm{eff}}R_{\mathrm{eff}}$ at different field strengths $b=|\bm{b}|$. This is defined in the same way as before, only now it depends on the ratio $\left[\kappa_{\mathrm{eff}}R_{\mathrm{eff}}-h_{\mathrm{mod}}\right]|\bm{b}|/\left[1+\kappa_{\mathrm{eff}}R_{\mathrm{eff}}h_{\mathrm{mod}}(\bm{b}\cdot\bm{b})\right]$.} \label{f:3}
\end{figure}

\begin{remark}

This example is also applicable to a material where mobile charges occupy complementary positions on the same lattice. For example, following \cite{bai:25} and \cite{kuganathan:21}, it is conceivable that the charges associated with available Li-vacancy sites in LLZO may be mobile independently of lithium via oxygen ion or electron diffusion, enabling behavior similar to that of a binary electrolyte but with electromagnetic-chemical potentials constrained as discussed above. Under this set of assumptions, the conduction model described above could be applicable to LLZO, which provides an explain the results of \cite{zhang:22}. See Figures \ref{f:2} and \ref{f:3} for a visual depiction of regimes where magnetic influence on conduction in a binary solid electrolyte can be significant. As demonstrated by \cite{CARLSON:26}, significant perpendicular reductions in conductivity lead to a desired effect of reduced current concentrations near imperfections.

    It should be cautioned that not all parameter combinations leading to a near degenerate $\bm{\rho}$ tensor will result in a significant observed effect, as it may still be the case that $\kappa_{\mathrm{eff}}R_\mathrm{eff}-h_\mathrm{eff}\approx 0$. Therefore, even if an effective, significant magneto-resistance in solid electrolyte materials is theoretically possible, it may be very rare or ephemeral. Observations such as those of \cite{Kim:25} and \cite{zhang:22} may then be mis-attributing the effect of the magnetic field to the conductivity, while the true effect of the magnetic field on transport may be due to its direct and indirect influence on the driving forces.
\end{remark}

\section{Application to the superionic conductor Pb$_{0.66}$Cd$_{0.34}$F$_2$}\label{s:5}

One of the few direct measurements of non-negligible magnetoresistance in an ionic conductor was recently reported by \cite{Yakushkin:25} for the fluoride conductor Pb$_{0.66}$Cd$_{0.34}$F$_2$. It is observed that the binary conduction model of Section \ref{s:4.2} is capable of fitting the experimental results significantly better than the classical single carrier model of \ref{s:3.1}. Given that Pb$_{0.66}$Cd$_{0.34}$F$_2$ allows for conduction via both interstitial diffusion and vacancy hopping, it is sensible to model the multi-component transport of counter-charged cation sites and interstitial F$^-$ ions. In a direction perpendicular to the magnetic field, it follows from \eqref{eq:keff_con} that the change in resistivity from the zero field value is
\begin{align*}
    \frac{\Delta\rho}{\rho}=\frac{\kappa_\mathrm{eff}^0R_\mathrm{eff}h_\mathrm{mod}\left(\frac{\kappa_\mathrm{eff}^0R_\mathrm{eff}}{h_\mathrm{mod}}-1\right)\bm{b}\cdot\bm{b}}{1+\kappa_\mathrm{eff}^0R_\mathrm{eff}h_\mathrm{mod}\bm{b}\cdot\bm{b}}
\end{align*}
under the additional, reasonable assumption that $R_h^{(1)}R_h^{(2)}\bm{b}\cdot\bm{b}$ is small compared to $\rho_D=\rho_{11}\rho_{22}-\rho_{12}^2$. A fit to the data of \citet[][Fig.~3]{Yakushkin:25} is shown in Figure \ref{f:4}, which clearly agrees with the data much better than the typical quadratic model. Hence, the binary conduction model proposed in this work serves as a potential explanation for the high magnetic field saturation of the magnetoresistance effect observed by \cite{Yakushkin:25} that would not be predicted by \eqref{eq:SIC} alone.

\begin{figure}\centering
\includegraphics[height=0.35\textheight]{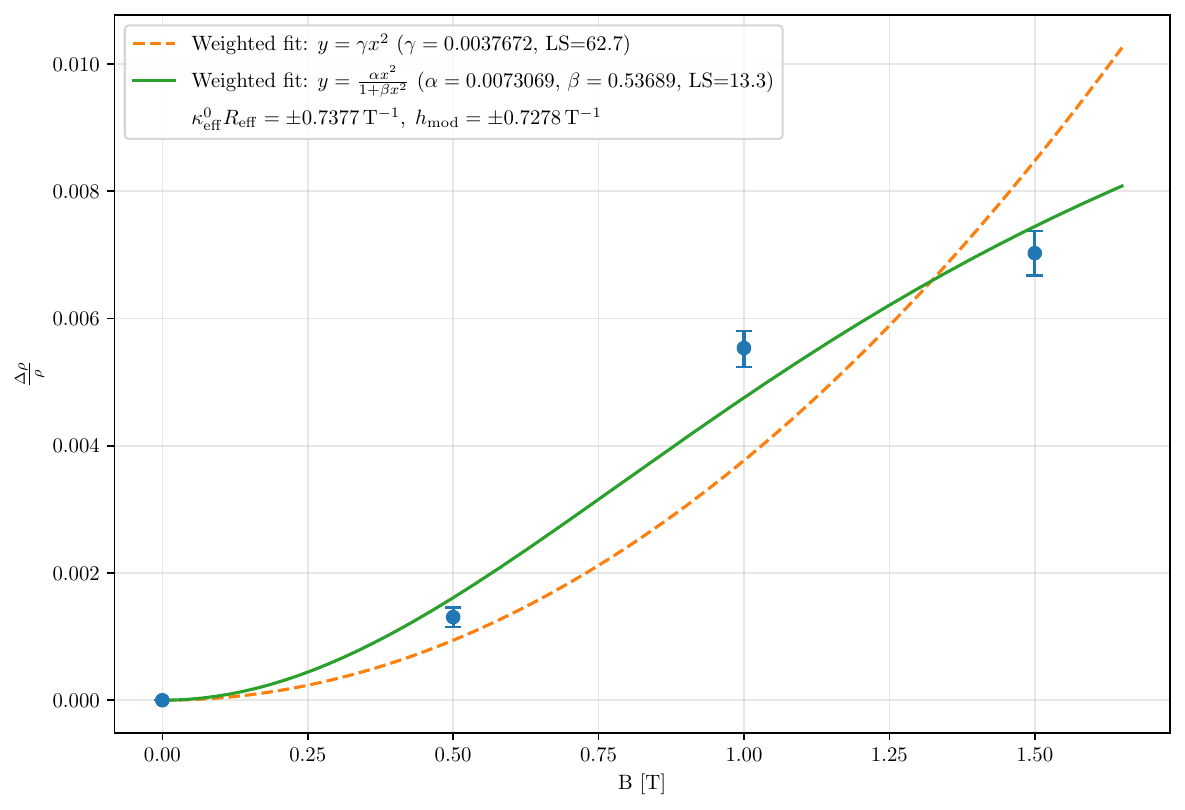}
\caption{Fits to the data from \citet[][Fig.~3]{Yakushkin:25} corresponding to the single-carrier model (dashed) and binary conductor model (solid), with weighted least-squares error given as LS in the legend. Here, $\gamma=\kappa R_h^2$ in the quadratic, single-carrier model, whereas $\alpha=\kappa_\mathrm{eff}^0R_\mathrm{eff}h_\mathrm{mod}\left(\frac{\kappa_\mathrm{eff}^0R_\mathrm{eff}}{h_\mathrm{mod}}-1\right)$ and $\beta=\kappa_\mathrm{eff}^0R_\mathrm{eff}h_\mathrm{mod}$ in the binary conductor model. Explicitly, LS is the sum of the square of the residuals, $y_i-\hat{y}_i$ divided by the experimental error at each point, $\sigma_i$: $LS=\sum_i(\frac{y_i-\hat{y}_i}{\sigma_i})^2$.}
\label{f:4}
\end{figure}

Apart from alignment with the results of \cite{Yakushkin:25}, the models proposed in Sections \ref{s:4.2} and \ref{s:3.3} yield relevant insights for other materials, particularly in cases with null results. For example, a Hall voltage was not detected in AgBr by \cite{liou:90}, despite the known defect mobility and an operationally measurable Hall mobility for similar materials like AgI. This may indicate that \eqref{eq:keff_con} can be used as an effective conductivity for the vacancies and interstitial defects in AgBr with $\kappa_{\mathrm{eff}}R_{\mathrm{eff}}-h_{\mathrm{mod}}\approx 0$. That is, the Hall effect contribution from one carrier likely cancels the contribution from the other in AgBr, while the same clearly cannot be said about AgI.

\begin{remark}
    Note that \cite{Yakushkin:25} observes a slow ramping of the magnetoresistive effect, and an even slower return to the 0 T reference result once the field is removed. However, \cite{Yakushkin:25} does note that temperature effects can be corrected for over the first 10 minutes of the experiment, and this slow de-saturation was observed only after significant heating was induced. Further, the 0.5\, T result is roughly constant over the first 10 minutes. Therefore, it is likely the case that field ramping past 0.5\,T was slow, and that measurements past 10 minutes were polluted by heating induced by Joule heating of the electromagnet. The latter of these two scenarios was implied by \cite{Yakushkin:25}, but the former must be assumed given a lack of operational data regarding the electromagnet. In any case, we presume that temperature effects were correctly accounted for, implying the effect is truly magnetoresitive as claimed by \cite{Yakushkin:25}.
\end{remark}

\section{Conclusions}

As demonstrated in Sections \ref{s:2}, \ref{s:3}, and \ref{s:4}, magnetic fields at achievable field strengths will have an appreciable influence on ion transport in solids via the influence of the Lorentz force alone if one or more of the following cases are applicable:
\begin{enumerate}
    \item the Hall parameter $R_h^{(i)}\bm{\kappa}_i$, where $\bm{\kappa}_i$ is the partial conductivity of species ${i}$, is of an order greater than or equal to 1\,T$^{-1}$. \label{case1}
    \item in the case of multi-component transport of charged species, the global transport tensor $\bm{\rho}$ is near degenerate, but $\kappa_{\mathrm{eff}}$ is well defined and experimentally measurable. \label{case2}
    \item the magnetic field meaningfully effects the modified drag tensors $\bm{M}_{ij}'$ by, for example, influencing the average momentum transfer during collisions or attempt frequency of hopping in a solid lattice. \label{case3}
\end{enumerate}
In cases of multi-component transport where at least one of the above criterion are satisfied, it is also necessary to check that Hall-like effects on carriers of opposite sign do not cancel if they are independently significant. In principle, the potential applicability of Case \ref{case2} in solid materials where coupled transport phenomena are not sufficiently understood may explain the seemingly anomalous response of materials such as LLZO to a magnetic field as demonstrated by \cite{zhang:22}. However, without rigorous experimental validation, it is difficult to definitively determine the true origin of the observed effect. For example, correlation and species interaction in solids may give rise to more magnetic influence on $\bm{M}_{ij}'$ than would be predicted using the approach of \cite{Kinouchi:1988} and \citet[][pp.199-201]{Balakrishnan:2021}.

Although experimental data relating to magnetic influence on ionic transport in solids is sparse, this work can serve as a useful tool in analytically determining the transport properties a material must have in order to determine whether a magnetic field's influence on transport is non-negligible. However, given that the magnetic field can influence transport indirectly by affecting other relevant responses, future work should also focus on a more multiphysics investigation of magnetic effects. Indeed, magnetic influence on the electromagnetic-chemical potential of species, mechanical deformation, or interfacial reaction rates may influence results such as those obtained by \cite{zhang:18} or \cite{Kim:25} more than is currently assumed. For example, in a soft, viscoelastic ionogel solid ionic conductor, a magnetic field may induce deformation that meaningfully influences the electromagnetic-chemical potentials of ions, which would constitute an indirect magnetic field effect on transport.

In any case, the lack of experimental work investigating the magnetoresistance of ions in solid electrolytes, along with existence of experiments that clearly confirm a nontrivial magnetic influence on ionic transport, suggests that many existing materials may already satisfy Cases \ref{case1}, \ref{case2}, or \ref{case3}. Thus, further experimental confirmation and exploration of the effect is critical in understanding both which materials may allow for ionic magnetoresistance and other ways in which a magnetic field can influence transport. In addition, existing continuum models almost never account for electrochemical-magnetic couplings, and those that do neglect other multiphysics couplings, such as stress and temperature, critical in understanding behavior of electrochemical systems \citep[see, e.g.,][]{CARLSON:26}. Therefore, the development of more comprehensive, multiphysics continuum models is clearly necessary to better understand the intricate ways in which an applied magnetic field can effect transport, especially given the evidence that such an effect exists and can be subtle yet substantial.

\bibliographystyle{elsarticle-harv}
\bibliography{FDMEMCT.bib}

\newpage

\backmatter
\begin{appendices}

\section{Algebraic properties of isotropic tensors}\label{a:a}

\begin{claim}[1]
    $\text{span}\{\bm{I},\bm{\epsilon}\cdot\bm{b},\bm{b}\otimes\bm{b}\}$ is a commutative ring \citep[see e.g.,][]{aluffi:09}. \label{lemma:1}
\end{claim}
\begin{proof}
Let $\text{span}\{\bm{I},\bm{\epsilon}\cdot\bm{b},\bm{b}\otimes\bm{b}\}=R_b$. Obviously, this is a vector space under tensor addition. It remains to show that the ring operation is commutative and that $R_b$ is closed under tensor single contraction (matrix multiplication). First, observe that
\begin{align*}
    \bm{\epsilon}\cdot\bm{b}(\bm{b}\otimes\bm{b})=(\bm{b}\times\bm{b})\otimes\bm{b}=\bm{0}
\end{align*}
and
\begin{align*}
    (\bm{b}\otimes\bm{b})\bm{\epsilon}\cdot\bm{b}=\bm{b}\otimes(\bm{\epsilon}\cdot\bm{b})^T\bm{b}=-\bm{b}\otimes(\bm{\epsilon}\cdot\bm{b})\bm{b}=\bm{b}\otimes(\bm{b}\times\bm{b})=\bm{0}.
\end{align*}
Second, see that for any vector field $\bm{c}\in\chi(\mathbb{R}^3)$ (the space of vector fields over $\mathbb{R}^3$),
\begin{align*}
    (\bm{\epsilon}\cdot\bm{b})(\bm{\epsilon}\cdot\bm{b})\bm{c}=(\bm{\epsilon}\cdot\bm{b})(\bm{c}\times\bm{b})=(\bm{c}\times\bm{b})\times\bm{b}=(\bm{b}\cdot\bm{c})\bm{b}-(\bm{b}\cdot\bm{b})\bm{c}=(\bm{b}\otimes\bm{b}-(\bm{b}\cdot\bm{b})\bm{I})\bm{c}.
\end{align*}
So, since $\bm{c}$ is arbitrary, we must have that $(\bm{\epsilon}\cdot\bm{b})(\bm{\epsilon}\cdot\bm{b})=\bm{b}\otimes\bm{b}-(\bm{b}\cdot\bm{b})\bm{I}\in R_b$. Obviously, $(\bm{b}\otimes\bm{b})^2=(\bm{b}\cdot\bm{b})(\bm{b}\otimes\bm{b})$ and $\bm{I}$ by definition commutes with all second order tensors. Therefore, $R_b$ is closed under multiplication and is commutative, and is thus a commutative ring. 
\end{proof}

\begin{claim}[2]
    Let $\bm{M}\in\text{span}\{\bm{I},\bm{\epsilon}\cdot\bm{b},\bm{b}\otimes\bm{b}\}$, and suppose $\bm{M}^{-1}$ exists. Then $\bm{M}^{-1}\in \text{span}\{\bm{I},\bm{\epsilon}\cdot\bm{b},\bm{b}\otimes\bm{b}\}$.\label{lemma:2}
\end{claim}
\begin{proof}
Suppose $\bm{M}\in \text{span}\{\bm{I},\bm{\epsilon}\cdot\bm{b},\bm{b}\otimes\bm{b}\}$. By the Cayley–Hamilton theorem \citep[see e.g.,][]{Mccarthy:33}, $\bm{M}^{-1}$ is a polynomial function of $\bm{M}$ and $\bm{I}$. Therefore $\bm{M}^{-1}\in\text{span}\{\bm{I},\bm{\epsilon}\cdot\bm{b},\bm{b}\otimes\bm{b}\}$ by Claim 1, as we must have $\bm{M}^n\in \text{span}\{\bm{I},\bm{\epsilon}\cdot\bm{b},\bm{b}\otimes\bm{b}\}$ for any $n$.
\end{proof}

\begin{claim}[3]
    Suppose $\bm{Q}$ is an invertible $2\times 2$ block symmetric tensor in $T(\chi(\mathbb{R}^3),\chi(\mathbb{R}^3))\times T(\chi(\mathbb{R}^3),\chi(\mathbb{R}^3))$ (the space of second order tensors fields over $\mathbb{R}^3$, cartesian product with itself). If each block is an element of the same commutative algebra and each block is invertible, then
    \begin{align*}
        \bm{Q}^{-1}=\begin{bmatrix}
            \bm{A} &\bm{B}\\
            \bm{B} & \bm{C}
        \end{bmatrix}^{-1}=\begin{bmatrix}
            (\bm{A}\bm{C}-\bm{B}^2)^{-1}\bm{C} & -(\bm{A}\bm{C}-\bm{B}^2)^{-1}\bm{B}\\
            -(\bm{A}\bm{C}-\bm{B}^2)^{-1}\bm{B} & (\bm{A}\bm{C}-\bm{B}^2)^{-1}\bm{A}
        \end{bmatrix},
    \end{align*}
    provided that each Schur block is invertible also.\label{lemma:3}
\end{claim}
\begin{proof}
    By the usual Schur complement inversion formula \citep[see, e.g.,][p.~14]{Zhang:05},
    \begin{align*}
        \begin{bmatrix}
            \bm{A} &\bm{B}\\
            \bm{B} & \bm{C}
        \end{bmatrix}^{-1}&=\begin{bmatrix}
            (\bm{A}-\bm{B}\bm{C}^{-1}\bm{B})^{-1} &-(\bm{A}-\bm{B}\bm{C}^{-1}\bm{B})^{-1}\bm{B}\bm{C}^{-1}\\
            -(\bm{C}-\bm{B}\bm{A}^{-1}\bm{B})^{-1}\bm{B}\bm{A}^{-1}& (\bm{C}-\bm{B}\bm{A}^{-1}\bm{B})^{-1}
        \end{bmatrix}\\
        &=\begin{bmatrix}
            [(\bm{A}\bm{C}-\bm{B}{^2})\bm{C}^{-1}]^{-1} &-(\bm{A}\bm{C}-\bm{B}^2)^{-1}\bm{B}\\
            -(\bm{A}\bm{C}-\bm{B}^2)^{-1}\bm{B}& [(\bm{A}\bm{C}-\bm{B}^2)\bm{A}^{-1}]^{-1}
        \end{bmatrix}\\
        &=\begin{bmatrix}
            (\bm{A}\bm{C}-\bm{B}^2)^{-1}\bm{C} & -(\bm{A}\bm{C}-\bm{B}^2)^{-1}\bm{B}\\
            -(\bm{A}\bm{C}-\bm{B}^2)^{-1}\bm{B} & (\bm{A}\bm{C}-\bm{B}^2)^{-1}\bm{A}
        \end{bmatrix}
    \end{align*}
    since the tensor multiplication commutes in this case and $\bm{A},\,\bm{B},\,\bm{C},\,(\bm{A}-\bm{B}\bm{C}^{-1}\bm{B}),\,(\bm{C}-\bm{B}\bm{A}^{-1}\bm{B})$ are invertible. Observe in particular that if the blocks commute with one another in this way then the inverse is precisely analogous to the inverse of a $2\times2$ matrix over $\mathbb{R}$.
\end{proof}

\end{appendices}

\end{document}